\documentclass{ws-ijfcs}

\usepackage{ulem}
\usepackage[T1]{fontenc}

\usepackage{amsmath}
\usepackage{amssymb}
\usepackage{comment}

\usepackage{color}
\usepackage{verbatim}

\usepackage{amsfonts}

\newcommand{\F}{ \ensuremath{ \mathbb{F}}}

\newcommand{\C}{ \ensuremath{ \mathbb C}}

\newcommand{\Z}{ \ensuremath{ \mathbb Z}}

\usepackage{graphicx}

\usepackage{hyperref}

\begin{document}

\markboth{Laurent Poinsot, Alexander Pott}{Non-Boolean Almost Perfect Nonlinear Functions 
on Non-Abelian Groups}

\title{Non-Boolean Almost Perfect Nonlinear Functions 
on Non-Abelian Groups}
\author{Laurent Poinsot}
\address{LIPN - UMR CNRS 7030, 
Institut Galil\'ee, University Paris XIII,\\ 
99, avenue Jean-Baptiste Clément,\\  
93430 Villetaneuse\\
\email{laurent.poinsot@lipn.univ-paris13.fr}}

\author{Alexander Pott}
\address{Department of Mathematics, Otto-von-Guericke-University
Magdeburg,\\ 39106 Magdeburg, Germany\\
\email{alexander.pott@ovgu.de}}

\maketitle

\begin{history}
\received{(Day Month Year)}
\accepted{(Day Month Year)}
\comby{(xxxxxxxxxx)}
\end{history}

\begin{abstract}
\noindent The purpose of this paper is to present the extended definitions 
and characterizations of the 
classical notions of APN and maximum nonlinear Boolean functions to deal with 
the  case of 
mappings from a finite group $K$ to another 
one $N$ with the possibility that one or both groups are non-Abelian.

\end{abstract}


\ccode{2000 Mathematics Subject Classification: 94A60, 20C05, 05B10}

\section{Highly nonlinear functions in the Abelian group setting}

The study of nonlinear properties of Boolean functions is one of 
the major tasks in secret-key cryptography. But the  
adjective "nonlinear" has several meanings: it can be related to the
resistance 
against the famous differential attack 
\cite{BS} and, in this interpretation, actually refers to (almost) perfect 
nonlinear functions. Moreover nonlinearity is also 
related to the maximum magnitude of the Fourier spectrum of Boolean functions 
- under the names 
"bent", "almost bent" or "maximal nonlinear" functions - which is itself
linked 
to the resistance 
against the linear attack \cite{Mat}. These two ways to define nonlinearity
 are not independent and even, in 
many situations, are exactly the same. 

Most of the studies and results on nonlinearity concerns Boolean 
functions, or in other words, functions 
from $K$ to $N$ where $K$ and $N$ are both elementary Abelian $2$-groups. Even 
if this kind of groups 
seems to be very natural for cryptographic purpose, there is no rule that
prevent 
us to use more 
complex groups and even non-Abelian ones. In this paper, we will
discuss the standard notion of nonlinearity in the 
non-Abelian setting. Nevertheless we begin by some definitions and basic tools 
used to study nonlinearity 
in the Abelian case. In order to keep the paper fairly selfcontained 
some proofs of well-known results will be added.

Let $K$ and $N$ be two finite groups written multiplicatively of 
orders $m$ and $n$ respectively. Let 
$G$ be the direct product $K \times N$. 
A mapping $f : K \rightarrow N$ is called 
{\bf{perfect nonlinear}}  (see~\cite{nyberg}) 
if and only if for each $a \in K$, $a \not = 1_K$ and each $b \in N$, the
 quantity 
\begin{equation}
\delta_f (a,b) = |\{g \in K | f(ag)(f(g))^{-1}=b\}|
\end{equation}
is constant equals to $\frac{m}{n}$. In \cite{CD04} the reader may find a very 
complete survey on the subject. 
In the cases where $n$ does not divide $m$, it is impossible for perfect 
nonlinear 
functions to exist. 
We note that perfect nonlinear functions also cannot exist if $K$ and
$N$ are elementary Abelian $2$-groups of the same order. Actually, a more
  general result holds:
\begin{theorem}
Let $K$ be an arbitrary group of order $m=2^a$, and let $N$ be an  
Abelian group of order $n=2^b$. A perfect nonlinear function $f:K\to N$
does not exist, if
\begin{arabiclist}
\item $a$ is odd;
\item $a=2s$ is even and $b\geq s+1$.
\end{arabiclist}
\end{theorem}

For proof, we refer to \cite{davis92a,ma-schmidt95,nyberg}.
Since perfect nonlinear functions do
not exist in many cases, the following definition is meaningful:
we call $f:K\to N$ an {\bf{almost perfect nonlinear}} (APN) 
function 
(see~\cite{nyberg93}) if and only if 
\begin{equation}\label{eq:apn}
\displaystyle \sum_{(a,b)\in G}\delta_f (a,b)^2 \leq \sum_{(a,b)\in G}\delta_g (a,b)^2\ \forall g:K \rightarrow N\ .
\end{equation}
Both these definitions do not use the commutativity in a group.
Hence these definitions also apply to the non-Abelian situation.

With each function $f : K \rightarrow N$ we associate its graph $
D_f \subseteq G$:
\begin{equation}
D_f = \{(g,f(g)) | g\in K\}\ .
\end{equation}
This set plays an important role in the study of nonlinear properties of the 
corresponding function. For instance $f$ is 
perfect nonlinear if and only if its graph is a splitting 
$(m,n,m,\frac{m}{n})$ difference set in $G$ relative to the 
normal subgroup $N$. Recall that a set $R \subseteq G=K\times N$ of cardinality $k$ 
is a (splitting) $(m,n,k,\lambda)$ difference set in 
$G$ relative to $N$ if and only if the following property holds:
the list of nonidentity quotients $r(r')^{-1}$ with $r,r' \in R$ covers every 
element in $G\setminus \{1_K\}\times N$ 
precisely 
$\lambda$ times and no element in $\{1_K\}\times (N\setminus\{1_N\})$ is
covered.
The term {\it splitting} refers to the fact that the group in which the
relative difference set exists is $K\times N$, hence it ``splits'',
where one of the factors is the ``forbidden subgroup''. We note
that also non-splitting relative difference sets are studied, however 
for applications in cryptography only {\it mappings} $f:K\to N$
seem to be of interest, and then the graph corresponds 
to a splitting relative difference sets. Moreover, we have
$k=m$ in this situation, and this case is called {\it semi-regular}. 
There are also many relative difference sets known where $k\ne m$. 
For a  survey on relative difference sets,
we refer to \cite{pott-survey}.

In general such combinatorial structures are studied using the notion of 
group algebras. 
Let $G$ be a finite group (written multiplicatively) and $R$ a 
commutative ring
with a unit. 
We denote by $R[G]$ the 
{\bf{group algebra}} of $G$; its underlying $R$-module is free and has a basis indexed 
by the elements of $G$ and 
which is identified to $G$ itself: so it is a free $R$-module of rank $|G|$
 and, as such, isomorphic to the direct sum 
$\displaystyle \bigoplus_{g\in G}R_g$ where for each $g \in G$, $R_g =R$.
 
Every element $D$ of $R[G]$ can be uniquely represented as 
\begin{equation}
D=\displaystyle \sum_{g \in G}d_g g,\ \mbox{with}\ d_g \in R\ .
\end{equation}
The addition in $R[G]$ is given as a component-wise addition of $R$. More precisely 
\begin{equation}
\displaystyle \left(\sum_{g \in G}a_g g\right)+\left(\sum_{g \in G}b_g
g\right) = 
\sum_{g \in G}
\left( a_g + b_g\right)g 
\end{equation}
while the multiplication - convolutional product - is  
\begin{equation}
\displaystyle \left(\sum_{g \in G}a_g g\right)\left(\sum_{g \in G}b_g g\right) 
= \sum_{g \in G}
\left(\sum_{h \in G}a_h b_{h^{-1}g}\right)g\ . 
\end{equation}
Finally the scalar multiplication by elements of $R$ is the usual one 
\begin{equation}
\displaystyle\lambda \left (\sum_{g \in G}d_g g\right) = \sum_{g \in G}(\lambda d_g)g,\ \mbox{with}\ 
\lambda \in R\ .
\end{equation}
Any subset $D$ of $G$ is naturally identified with the following element of
 $R[G]$
\begin{equation}
\displaystyle \sum_{g \in G}1_D(g)g = \sum_{g \in D}g\ ,
\end{equation}
where we define the {\bf{indicator function}} of $D$ 
\begin{equation}
1_D(g) = \left \{
\begin{array}{ll}
1_R&\mbox{if}\ g\in D\ ,\\
0_R&\mbox{if}\ g\not \in D\ . 
\end{array}
\right .
\end{equation}
When $R = \mathbb{C}$ we define 
\begin{equation}
\displaystyle \left (\sum_{g \in G}d_g g\right )^{(-1)} = \sum_{g \in G}\overline{d_g} g^{-1} = 
\sum_{g \in G}\overline{d_{g^{-1}}}g
\end{equation}
where $\overline{z}$ is the complex conjugate of $z \in \mathbb{C}$.\\
\noindent For instance for a function $f : K \rightarrow N$ and $G=K\times N$, 
we have 
\begin{equation}
D_f D_f^{(-1)}=\displaystyle \sum_{(a,b)\in G}\delta_f (a,b)(a,b) \in \mathbb{Z}[G]\ .
\end{equation}

As we claimed  above, $(m,n,k,\lambda)$ difference sets in $G
=K\times N$ relative to $N$ 
have a natural interpretation in $\mathbb{Z}[G]$ because $R \subseteq G$ is
 such a set if and only if it satisfies the 
following group ring equation:
\begin{equation}\label{eq:rds}
RR^{(-1)}=k1_G + \lambda (G - N)\ .
\end{equation}
Therefore, we have the following well known theorem which holds in the
Abelian as well as in the non-Abelian case:

\begin{theorem}
A function $f:K\to N$ is perfect nonlinear if and only if
the graph $D_f$ of $f$ is a splitting 
$(m,n,m,\frac{m}{n})$ difference set relative to $\{1_K\}\times N$.
\end{theorem}

Another important tool for the study of highly nonlinear mappings
 - but restricted to the case of finite 
{\bf{Abelian}} groups - is the notion of group characters. A {\bf{character}} 
$\chi$ of a finite Abelian group is a group homomorphism from 
$G$ to the multiplicative group $\C^\ast$ of
$\C$. The elements $\chi(g)$ belong to the unit circle of 
$\C$.

The set of all such characters of a given Abelian group $G$, when equipped with the 
point-wise multiplication of mappings, 
is itself a group (called the {\bf{dual group}}, denoted by $\widehat{G}$), 
isomorphic to $G$. The character $\chi_0 : g \in G \mapsto 1$ 
is called the {\bf{principal}} character of $G$. The characters of a direct
 product 
$K \times N$ are given by $\chi = \chi_K \otimes \chi_N$
where $(a,b)\in K \times N$ is mapped to 
  $\chi_K (a)\chi_N (b)\in \mathbb{C}$, and where $\chi_K$ is a character of $K$ and $\chi_N$ a character of $N$. 
The characters of $G$ 
can be naturally extended by linearity to homomorphisms of algebras from $\mathbb{C}[G]
$ to $\mathbb{C}$: if
$D=\displaystyle \sum_{g\in G}d_g g \in \mathbb{C}[G]$ and $\chi$ is a
character of 
$G$, then 
\begin{equation}
\chi(D)=\displaystyle \sum_{g \in G}d_g \chi(g)\ .
\end{equation}  

There is an important formula for characters of Abelian groups
which also holds for non-Abelian groups, see Theorem 
\ref{theo:inv-par-na}.

\begin{theorem}[Inversion formula]
Let $G$ be an Abelian group, and let $D=\sum_{g\in G} d_g g$
be an element in $\C[G]$. Then
$$
d_g=\frac{1}{|G|} \sum_{\chi\in\hat{G}}\chi(D)\cdot \chi(g^{-1}).
$$
\end{theorem}

\begin{corollary}[Parseval's equation]\label{cor:parseval_abelian}
Let $G$ be an Abelian group. For $D=\sum_{g\in G} d_g g$ in
$\C[G]$, we have
$$
\sum_{g\in G} d_g^{\; 2} = \frac{1}{|G|}\sum_{\chi\in\hat{G}} |\chi(D)|^2.
$$
\end{corollary}
\begin{proof} This follows easily from the inversion formula applied to the
coefficient of the identity element in $D\cdot D^{-1}$.
\end{proof}

Using these characters we can introduce another criterion of nonlinearity: 
a function $f : K \rightarrow N$ (where 
$K$ and $N$ are both finite and Abelian) is called 
{\bf{maximum nonlinear}} if and only if 
\begin{equation}
\max_{\chi_N\not=\chi_0}|\chi(D_f)| \leq
 \max_{\chi_N \not=\chi_0}|\chi(D_g)|\quad  \forall g : K \rightarrow N
\end{equation}
or,  equivalently
\begin{equation}
\max_{\chi_N\not=\chi_0}|\chi(D_f)| = \min_{g:K\rightarrow N}\max_{\chi_N \not=\chi_0}|\chi(D_g)|\ .
\end{equation}

The value $\sqrt{|K|}$ is a lower bound for the quantity 
$\displaystyle \max_{\chi_N\not=\chi_0}|\chi(D_f)|$ which follows easily from Parseval's relation for Abelian groups (Corollary \ref{cor:parseval_abelian}). 
A function that reaches this theoretically best bound is called {\bf bent}. A function is bent if and only if it is   perfect nonlinear as defined above.. This follows easily by applying characters to the group ring equation (\ref{eq:rds}), see also \cite{CD04,pott}.
 
Finally, characters allow us to give another characterization for APN functions:
\begin{theorem}\label{theo:fourthpowers}
Let $K$ and $N$ be two finite Abelian groups. A function $f : K \rightarrow
 N$ is almost perfect nonlinear if and only if 
\begin{equation}
\displaystyle \sum_{\chi}|\chi(D_f)|^4 \leq \sum_{\chi}|\chi(D_g)|^4\ \forall
g : K \rightarrow N\ .
\end{equation}
\end{theorem}
\begin{proof} ({\it sketch}) As before, let $G=K\times N$.
Consider the coefficient of the identity
    element
in $(D_gD_g^{(-1)})^2$: this is precisely $\sum_{(a,b)\in G}\delta_g(a,b)^2$..
Therefore, minimizing $\sum_{(a,b)\in G}\delta_g(a,b)^2$ is equivalent to
minimizing $\sum_{\chi}|\chi(D_g)|^4$.
\end{proof}

\begin{remark}
\begin{enumerate}
\item If all character values $\chi(D_f)$ with $\chi_N\ne \chi_0$ are the same, then 
an almost perfect nonlinear function is actually perfect nonlinear.
\item Theorem \ref{theo:fourthpowers} is well known for the
elementary Abelian case (see \cite{chabaud}) and also known for the
Abelian case \cite{pott}. It is one of the
purposes of this paper to show that one can even extend it to the
non-Abelian case.
\end{enumerate}
\end{remark}

In some particular cases, we know a lower bound for the sum of the 
fourth-power of the absolute value of $\chi(D_f)$. 
Indeed when 
$K$ and $N$ are two elementary Abelian $2$-groups of order $m$, then it can
 be shown that for each $f : K \rightarrow N$, 
\begin{equation}\label{fourth_pow_1}
\displaystyle 2m^3(m-1)\leq \sum_{\chi\not=\chi_0} |\chi(D_f)|^4
\end{equation} 
and since $\chi_0(D_f)=|D_f|=|K|=m$, we have 
\begin{equation}\label{fourth_pow_2}
\displaystyle m^3 (3m-2)\leq \sum_{\chi}|\chi(D_f)|^4\ .
\end{equation}
It is very important to note that these arguments only work in the elementary
Abelian case: the proofs rely on the fact that the $\delta_f(a,b)$ are always 
even in characteristic 2.
We refer the reader to the original paper \cite{chabaud}, 
see also \cite{CD04,pott}. Together with Corollary \ref{cor:parseval_abelian},
 inequality~(\ref{fourth_pow_1}) yields yet another lower bound
 for the quantity $\displaystyle \max_{\chi_N\not=\chi_0}|\chi(D_f)|$ in the 
elementary Abelian case:
\begin{theorem}[\cite{chabaud}]
Let $f : K \rightarrow N$ where $K$ and $N$ are elementary Abelian $2$-groups 
of order $m=2^n$. Then 
\begin{equation}
\displaystyle \max_{\chi_N\not=\chi_0}|\chi(D_f)|\geq \sqrt{2m}\ .
\end{equation}
Moreover when $n$ is odd, a  function $f : K \rightarrow N$ is maximal nonlinear if and only if  
\begin{equation}
\displaystyle \max_{\chi_N\not=\chi_0}|\chi(D_f)|=\sqrt{2m}\ .
\end{equation} 
\noindent In this case, the function is almost perfect nonlinear and the lower bound given in 
inequality~(\ref{fourth_pow_2}) is reached.
\end{theorem}

Contrary to the odd case, when $n$ is even obviously the previous lower bound cannot be reached. 
The lowest possible value 
of $\displaystyle \max_{\chi_N\not=\chi_0}|\chi(D_f)|$ satisfies (see, for
 instance \cite{carlet-survey-vectorial})
\begin{equation}\label{eq:m_even}
\displaystyle \sqrt{2m}<\min_{f:K\rightarrow N} \max_{\chi_N\not=\chi_0}|
\chi(D_f)|  \leq 2\sqrt{m}\ .
\end{equation} 
Moreover, the lowest known value corresponds to the upper bound. 

\begin{example}
\begin{enumerate}
\item The classical example of an almost perfect nonlinear 
function $f:\F_2^{\; n}\to \F_2^{\; n}$ is $f(x)=x^3$,
where the group is identified with the additive group of the
finite field $\F_{2^n}$. This function is also maximum nonlinear
if $n$ is odd. In the $n$ even case, it is a function where the
largest nontrivial  character value is $2^{\frac{n+2}{2}}$, hence
the upper bound in (\ref{eq:m_even}) is reached, see
\cite{carlet-survey-vectorial}, for instance. 
\item Similarly, the classical example of a perfect nonlinear mapping
is $x^2$ defined on $\F_{p^n}$, $p$ odd (this is folklore).
\end{enumerate} 
\end{example}

We note that recently many new perfect and almost perfect nonlinear
functions have been discovered. It seems that the ideas to construct 
perfect and almost perfect nonlinear functions are somewhat
similar, therefore the discovery of new almost perfect nonlinear
functions did influence the investigation of perfect nonlinear
functions which are of great interest in finite geometry. For the
characteristic $2$ case, we refer to \cite{carlet-survey-vectorial}
and the references cited there, in particular 
\cite{mcguire-ffa-2008,Budaghyan07D,Budaghyan06C,ekp,edel-amc};
for the perfect nonlinear case see 
\cite{BierBrauer2009-2,B-H:2011,B-H:2008,Zha2009}, for instance.  

The purpose of this paper is to develop  these
ideas - almost perfect nonlinearity and 
maximum nonlinear functions - in a larger context than the classical
elementary 
Abelian $2$-groups, namely the case 
of finite non-Abelian groups. For this objective we introduce the appropriate 
notion of "non-Abelian characters" in the following section.

\section{Basics on group representations}

\noindent The reader may find the definitions and results listed below (and
 lot more) in the book 
\cite{serre}.

A {\bf{(linear) representation}} of a finite group $G$ is a pair 
$(\rho,V)$ where $V$ is 
a finite-dimensional $\mathbb{C}$-vector space and $\rho$ is a group 
homomorphism from $G$ to the linear 
group $\mathit{GL}(V)$. In practice the group homomorphism $\rho$ is often 
identified with the representation 
$(\rho,V)$ and we may sometimes use the notation $V_{\rho}$ to refer to the
 second component of $(\rho,V)$. The {\bf{dimension}} of the representation
 $(\rho,V)$ is defined as the dimension of the 
vector space $V$ over the complex field and denoted by $\dim \rho$. Two 
representations 
$(\rho_1,V_1)$ and $(\rho_2,V_2)$ of the same group $G$ are said to be 
{\bf{equivalent}} 
if there exists a 
vector space isomorphism $T : V_1 \rightarrow V_2$ such that for each $g \in G$,
\begin{equation}
T \circ \rho_1 (g) = \rho_2 (g) \circ T\ .
\end{equation} 
Two equivalent representations may be identified without danger. 
Continuing with definitions, a representation $(\rho,V)$ is {\bf{irreducible}} 
if the only 
subvector spaces $W$ of $V$ so that $(\rho(g))(W)\subseteq W$ for every $g \in
G$ 
are the null-space 
 and 
$V$ itself. The {\bf{dual}} $\widehat{G}$ of a finite group $G$ is the set 
of all equivalence classes 
of irreducible representations of $G$. When $G$ is an Abelian group,
$\widehat{G}$ 
is  
dual group of $G$, as introduced
earlier. In particular a finite group is Abelian if and only if all 
its irreducible 
representations are one-dimensional. Let $[G,G]$ be the derived subgroup of
 $G$, {\it i.e.}, the subgroup generated 
by the elements of the form $ghg^{-1}h^{-1}$ for $(g,h) \in G^2$. It is a 
normal subgroup and the quotient-group $G/[G,G]$ is 
Abelian. One-dimensional representations of $G$ are related to this derived
 group since their total number is equal to 
the order of $G/[G,G]$. 
The cardinality $|\widehat{G}|$ of the dual of $G$ is equal to the number 
of conjugacy classes of $G$. 
The (equivalence class of the) representation $(\rho_0,\mathbb{C})$ that 
maps each element of $G$ to the 
identity map of $\mathbb{C}$ is called the {\bf{principal}} representation 
of $G$. For practical 
purpose we identify 
the value $\rho_0(g)$ as the number $1$ rather than the identity mapping of
 $\mathbb{C}$ or, in 
other terms, we identify $\rho_0 (g)$ and its trace $\mathit{tr}(\rho_0 (g))
=1$. 

We can also construct some linear representations from existing ones. 
For instance, given the dual sets of two 
finite groups $K$ and $N$, one can build up the dual $\widehat{K \times N}$
 of their direct product. This construction 
uses the notion of tensor product of two vector spaces. So let $V_1$ and $V_2$ 
be two complex vector spaces. We call 
{\bf{tensor product}} of $V_1$ and $V_2$ a vector space $W$ equipped with a
 bilinear mapping 
$j : V_1 \times V_2 \rightarrow W$ such that the following property holds. 
For each $\mathbb{C}$-vector space $V_3$ and for each bilinear mapping $f :
 V_1 \times V_2 \rightarrow V_3$ there exists 
one and only one linear mapping $\tilde{f} : W \rightarrow V_3$ such that 
$\tilde{f}\circ j = f$. 
It can be shown that there is one and only one such vector space $W$ (up to
 isomorphism). It is denoted by
$V_1 \otimes V_2$. Moreover if $(e^{(1)}_{k})_k$ is a basis of $V_1$ and 
$(e^{(2)}_{\ell})_{\ell}$ is a basis of $V_2$, then the family 
$(j(e^{(1)}_k,e^{(2)}_{\ell}))_{(k,\ell)}$ is a basis of $W$. This property shows that 
\begin{equation}
\dim (V_1 \otimes V_2) = \dim(V_1) \dim(V_2)\ .
\end{equation}
In particular for every vector space $V$, we have $V \otimes \mathbb{C} =
 \mathbb{C} \otimes V = V$ (equality up to 
isomorphism). 
For $(v_1,v_2)\in V_1 \times V_2$, we denote $j(v_1,v_2)$ by 
$v_1 \otimes v_2$. Since $j$ (and therefore $\otimes$) is 
a bilinear mapping we have $(\alpha v_1 + \beta w_1) \otimes v_2=\alpha  v_1 
\otimes v_2 + \beta w_1 \otimes v_2$ and the 
corresponding equality holds for the second variable. So now let 
$\rho_1 : K \rightarrow \mathit{GL}(V_1)$ and $\rho_2 : N \rightarrow 
\mathit{GL}(V_2)$ be two representations. Then we 
define a representation $\rho_1 \otimes \rho_2$ of $K \times N$ in $V_1 \otimes V_2$ by 
\begin{equation}
(\rho_1 \otimes \rho_2)(k,n) = \rho_1 (k) \otimes \rho_2
  (n)\quad  \mbox{with}\ (k,n) 
\in K \times N\ .
\end{equation}
Moreover, one can show that if $\rho_1$ and $\rho_2$ are both irreducible then 
$\rho_1 \otimes \rho_2$ is an irreducible 
representation of $K \times N$. And reciprocally, every irreducible
representation 
of $K \times N$ is equivalent to a 
representation of the form $\rho_1 \otimes \rho_2$, where $\rho_1$ 
(resp. $\rho_2$) is an irreducible representation of 
$K$ (resp. $N$). Given an irreducible representation $\rho$ of $K \times N$, 
we use $\rho_K$ and $\rho_N$ to denote 
the representations of $K$ and $N$ such that $\rho$ is equivalent to 
$\rho_K \otimes \rho_N$. A system of representatives 
of equivalence classes of irreducible representations of $K \times N$ is given 
by this way. Actually, the classes of representations of a given
  group $G$ may be used to form a ring: 
just take the free Abelian group generated by all isomorphism classes of
representations of $G$, 
mod out by the subgroup generated by elements of the form
$\rho_1+\rho_2-(\rho_1\oplus \rho_2)$, 
where $\rho_1\oplus \rho_2$ is the obvious direct sum of two (classes of)
representations. 
It can be proved that the irreducible representations form a basis 
for this $\mathbb{Z}$-module. The ring structure, called the {\bf
  representation ring of $G$}, 
is then given simply by tensor product
 $\rho_1\otimes\rho_2: g\in G \mapsto \rho_1(g)\otimes\rho_2(g)\in
 \mathit{GL}(V_1\otimes V_2)$, 
defined on these generators and extended by linearity.

Note that every irreducible 
representation is equivalent to a unitary representation {\it i.e.} a 
representation $\rho$ such that $\rho(g^{-1}) = 
\rho(g)^*$, where the star denotes the usual adjoint operation. Indeed, 
let  $\rho: G \to \mathit{GL}(V)$ be a representation of a finite group $G$, 
where $V$ is an Hermitian space (together with a scalar product
$\langle\cdot,\cdot\rangle$) 
and let us consider the following Hermitian product: 
$\langle x,y\rangle_{\rho}=\displaystyle\sum_{g\in
  G}\langle\rho(g)(x),\rho(g)(y)\rangle$. 
It is easy to prove that for every $x,y\in V$ and every $g\in G$, 
$\langle \rho(g)(x),\rho(g)(y)\rangle_{\rho}=\langle x,y\rangle_{\rho}$ in
such a 
way that $\rho(g)$ is unitary with respect to $\langle
\cdot,\cdot\rangle_{\rho}$ for each 
$g\in G$.  Therefore, in what follows we assume that $\widehat{G}$ 
is actually a complete set of representatives of non-isomorphic irreducible
representations, 
all of them being unitary.
 
We naturally extend a representation $(\rho,V)$  
by linearity to an algebra homomorphism from $\mathbb{C}[G]$ to 
$\mathit{End}(V)$, the linear endomorphisms of $V$, by
\begin{equation}
\rho (D) = \displaystyle \sum_{g\in G}d_g \rho(g)\ ,
\end{equation} 
with $D= \displaystyle \sum_{g\in G}d_g g$. As a particular case one can 
prove the following relations:
\begin{equation}\label{kind_of_orthogonality_relations}
\rho(G) = \left \{
\begin{array}{ll}
0_V &\mbox{if}\ \rho \not = \rho_0\ ,\\
|G| & \mbox{if}\ \rho =\rho_0
\end{array}\right .
\end{equation}
for every $\rho \in \widehat{G}$, and
\begin{equation}\label{kind_of_orthogonality_relations_2}
\displaystyle \sum_{\rho \in \widehat{G}}\dim \rho\ \mathit{tr}(\rho(g)) 
= \left \{
\begin{array}{cl}
|G| & \mbox{if}\ g = 1_G\ ,\\
0 & \mbox{if}\ g \not = 1_G\ .
\end{array}
\right .
\end{equation}
\noindent For $D \in \mathbb{C}[G]$ we may define its {\it Fourier transform} 
as 
$\widehat{D} = \displaystyle \sum_{\rho \in \widehat{G}}\rho(D)\rho$,
  which can be identified with $(\rho(D))_{\rho\in \widehat{G}}\in \displaystyle\bigoplus_{\rho\in\widehat{G}}\mathit{End}(V_{\rho})$. 
Then the 
Fourier transform is the mapping 
\begin{equation}
\begin{array}{llll}
\mathcal{F}:& \mathbb{C}[G]&\rightarrow & \displaystyle\bigoplus_{\rho
  \in\widehat{G}}
\mathit{End}(V_{\rho})\\
& D & \mapsto & \widehat{D}\ .
\end{array}
\end{equation}
Note that in the case where $G$ is a finite Abelian group, then 
every irreducible representation is one-dimensional and
$\mathit{End}(\mathbb{C})$ 
is 
isomorphic to $\mathbb{C}$. Therefore for $D \in \mathbb{C}[G]$, we have 
$\displaystyle \widehat{D}\in \bigoplus_{\rho \in \widehat{G}}\mathbb{C}\simeq 
\mathbb{C}[\widehat{G}]$.

The coefficients $\rho(D)$ of the Fourier transform $\widehat{D}$
 of $D$ are used in the sequel to 
study nonlinear properties of functions $f : K \rightarrow N$ in the general 
case where both $K$ and $N$ are finite groups, 
not necessarily commutative.

\section{Almost perfect nonlinearity} 

In this section, we develop a Fourier characterization of APN
functions 
in the non-Abelian setting, 
using group representations. We also introduce the relevant notion of bentness 
in this context and we 
show that, contrary to the classical case, this is not equivalent to
 perfect nonlinearity in this 
general framework.

Let $G$ be a finite group. The following theorem
is well known; we include a proof to keep the paper more self-contained.
\begin{theorem}[Fourier inversion, Parseval's equation]\label{theo:inv-par-na}
Let $D=\displaystyle \sum_{g\in G}d_g g$ be an element in the group algebra 
$\mathbb{C}[G]$. Then the following holds:
\begin{equation}
\displaystyle d_g = \frac{1}{|G|}\sum_{\rho \in \widehat{G}}\dim \rho\ 
\mathit{tr}(\rho(D)\circ \rho(g^{-1}))\quad 
(\mbox{Fourier\ inversion}) 
\end{equation}
\begin{equation}
\displaystyle \sum_{g \in G} |d_g|^2 = \frac{1}{|G|}\sum_{\rho\in\widehat{G}}
\dim \rho\ \left\|\rho(D)\right\|^2 
\quad (\mbox{Parseval's equation})
\end{equation}
where $\|f\|$ is the {\bf{trace norm}} of a linear endomorphism $f$ given by 
$\|f\|=\sqrt{\mathit{tr}(f\circ f^{*})}$. 
\end{theorem}
\begin{proof}
We have 
\begin{equation}
\begin{array}{lll}
& & \displaystyle \sum_{\rho \in \widehat{G}}\dim \rho\ \mathit{tr}(\rho(D)\circ
\rho(g^{-1})) \\
& = & \displaystyle \sum_{\rho \in \widehat{G}}\mathit{tr}\left(\rho( \sum_{h \in
  G}d_h h) 
\circ \rho(g^{-1})\right)\\
&=&\displaystyle \sum_{h \in G}d_h \sum_{\rho \in \widehat{G}}\dim \rho\ 
\mathit{tr}(\rho(hg^{-1}))\ 
\mbox{(by\ linearity\ of\ the\ trace\ and\ $\rho$)}\\
&=&|G|d_g\ \mbox{(according to eq.~(\ref{kind_of_orthogonality_relations_2}))}.
\end{array}
\end{equation}
This proves the Fourier inversion formula. 

We have 
\begin{equation}
DD^{(-1)} =\displaystyle \sum_{g \in G}\left(\sum_{h \in G}d_h \overline{d_{g^{-1}h}}\right )g\ .
\end{equation}
In particular the coefficient of the identity $1_G$ in this formal sum is 
$\displaystyle \sum_{g\in G}|d_g|^2$. We can also compute this coefficient 
by using the Fourier inversion on $DD^{(-1)}$. 
It is given by 
\begin{equation}
\begin{array}{lll}
&&\displaystyle \frac{1}{|G|}\sum_{\rho \in \widehat{G}} \dim \rho\
\mathit{tr}(\rho(DD^{(-1)}))\\
 &=&
\displaystyle \frac{1}{|G|}\sum_{\rho \in \widehat{G}}\dim \rho\ 
\mathit{tr}(\rho(D)\circ\rho(D)^{*}))\ 
(\mbox{since\ $\rho$\ is\ unitary})\\
&=&\displaystyle \frac{1}{|G|}\sum_{\rho \in \widehat{G}}\dim \rho \|\rho(D)\|^2\ .
\end{array}
\end{equation}
We may assume that $\widehat{G}$ is a set of unitary representatives of irreducible representations.
Therefore Parseval's equation holds.
\end{proof}

Using group representations we obtain an alternative formulation 
for almost perfect 
nonlinearity. Note that the definition of
almost perfect nonlinearity given in eq. (\ref{eq:apn})
holds for the Abelian as well as non-Abelian situation. 
\begin{theorem}\label{APN_characterization}
Let $K$ and $N$ be two finite groups. Let 
$G$ be the direct product $K \times N$. 
A function $f : K \rightarrow N$ is almost perfect nonlinear if and 
only if 
\begin{equation}
\displaystyle \sum_{\rho \in \widehat{G}}\dim \rho\ \|\rho(D_f)\|^4 \leq 
\sum_{\rho \in \widehat{G}}
\dim \rho\ \|\rho(D_g)\|^4,\ \forall g : K \rightarrow N\ .
\end{equation}
\end{theorem}
\begin{proof}
We have $\displaystyle D_f D_f^{(-1)} = \sum_{(a,b)\in G}\delta_f(a,b)(a,b)$. 
So using Parseval's equation we obtain: 
\begin{equation}
\begin{array}{lll}
\displaystyle \sum_{(a,b)\in G}\delta_f (a,b)^2 &=&\displaystyle \frac{1}{|G|}
\sum_{\rho \in \widehat{G}}
\dim\rho\ \|\rho(D_f D_f^{(-1)})\|^2\\
&=&\displaystyle \frac{1}{|G|}\sum_{\rho \in \widehat{G}}\dim\rho\ \|\rho(D_f)\circ \rho(D_f)^{*}\|^2\\
&=&\displaystyle \frac{1}{|G|}\sum_{\rho \in \widehat{G}}\dim\rho\ \|\rho(D_f)\|^4\ .
\end{array}
\end{equation}
Because a function $f : K \rightarrow N$ is APN if and only if for every 
$g : K \rightarrow N$, 
\begin{equation}
\displaystyle \sum_{(a,b)\in G}\delta_f (a,b)^2 \leq \sum_{(a,b)\in G}\delta_g (a,b)^2,
\end{equation} 
this concludes the proof.
\end{proof}
\noindent Group representations can be used to find another criterion for the 
nonlinearity of functions. As before, $K$ and 
$N$ are both finite groups of 
order $m$ and $n$, and $f : K \rightarrow N$. 
For some $\rho \in \widehat{K\times N}$, 
the values of $\rho(D_f)$  are known:
\begin{equation}\label{know_values_of_rho_Df}
\rho(D_f) = \left \{
\begin{array}{ll}
m & \mbox{if}\ \rho=\rho_0\ ,\\
0_{V} &\mbox{if}\ \rho=\rho_K \otimes
\rho_0\ \mbox{and}\ (\rho_K,V)\ \mbox{is\ nonprincipal on}\ K\ .
\end{array}
\right .
\end{equation}
\noindent Indeed first let us suppose that $\rho$ is principal on $G$. Then we
 have 
\begin{equation}
\begin{array}{lll}
\rho_0 (D_f)&=&\displaystyle \sum_{(a,b)\in G}1_{D_f}(a,b)\rho_0(a,b)\\
&=&\displaystyle \sum_{(a,b)\in G}1_{D_f}(a,b)\\
&=&|D_f|\\
&=&|K|\\
&=&m\ .
\end{array}
\end{equation}
\noindent Now let us suppose that $\rho = \rho_K \otimes \rho_0$ with
$(\rho_K,V)$ 
nonprincipal on $K$. Then we have 
\begin{equation}
\begin{array}{lll}
\rho(D_f) &=&\displaystyle \sum_{(a,b)\in G}1_{D_f}(a,b)\rho_K (a)\otimes
 \rho_0 (b)\\
&=&\displaystyle \sum_{a \in K}\rho_K (a)\otimes \rho_0 (f(a))\\
&=&\displaystyle \sum_{a \in K}\rho_K (a)\ (\mbox{since\ $V \otimes \mathbb{C} = V$})\\
&=&\rho_K (K)\\
&=&0_V\ (\mbox{according to eq.~(\ref{kind_of_orthogonality_relations})})
..
\end{array}
\end{equation}
\noindent Parseval's equation and an analogy with the Abelian case, suggest
 us to say that a function 
$f : K \rightarrow N$ is 
called {\bf{maximum nonlinear}} if and only if 
the value $\sqrt{\dim\rho} \|\rho(D_f)\|$ is as small as possible, or in other terms (using 
the known values of $\rho(D_f)$)
\begin{equation}\label{eq:maxnon}
\displaystyle \max_{\rho_N \not = \rho_0}\sqrt{\dim\rho}\ \|\rho(D_f)\|\leq 
\max_{\rho_N \not=\rho_0}
\sqrt{\dim\rho}\ \|\rho(D_g)\|\ 
\forall g:K \rightarrow N\ .
\end{equation}
As in the 
Abelian case, we obtain the following lower bound for this quantity:
\begin{theorem}
Let $f : K \rightarrow N$. Then 
\begin{equation}
\displaystyle \max_{\rho_N \not = \rho_0} \dim\rho\ \|\rho(D_f)\|^2 \geq \frac{m^2 (n-1)}{|\widehat{K}|(|\widehat{N}|-1)}\ .
\end{equation}
\end{theorem}
\begin{proof}
By Parseval's equation applied to $D_f$, we have
\begin{equation}
\begin{array}{lll}
\displaystyle \frac{1}{|G|}\sum_{\rho \in \widehat{G}}\dim\rho\ \|\rho(D_f)\|^2 &=&
\displaystyle \sum_{(a,b)\in G}1_{D_f}(a,b)^2\\
&=&m\ .
\end{array}
\end{equation}
So we have 
\begin{equation}\label{eq1}
\displaystyle \sum_{\rho \in \widehat{G}}\dim\rho\ \|\rho(D_f)\|^2=|G|m=m^2 n\ .
\end{equation}
We know some values of 
$\rho(D_f)$ that allow us to compute the following sum:
\begin{equation}
\begin{array}{lll}
\displaystyle \sum_{\rho_N \not = \rho_0}\dim\rho\ \|\rho(D_f)\|^2 &=& 
\displaystyle 
\sum_{\rho \in \widehat{G}}\dim\rho\ \|\rho(D_f)\|^2 - \sum_{\rho_N = \rho_0}\dim\rho\ \|\rho(D_f)\|^2\\
&=&\displaystyle m^2 n - \underbrace{\dim\rho_0\ \|\rho_0 (D_f)\|^2}_{=m^2} - 
\underbrace{\sum_{\rho_N = \rho_0,\ \rho\not=\rho_0}\dim\rho\ 
\|\rho(D_f)\|^2}_{=0}\\
&&(\mbox{according\ to\ eq.~(\ref{eq1})\ and\ eq.~(\ref{know_values_of_rho_Df})})\\
&=&m^2(n-1)\ .
\end{array}
\end{equation}
Now we need to evaluate the number of principal representations on $N$: it 
is equal to $|\widehat{K}|$. Then there are 
$|\widehat{G}|-|\widehat{K}|$ nonprincipal representations on $N$. But 
$|\widehat{G}|=|\widehat{K}||\widehat{N}|$. 
Therefore we have
\begin{equation}
\displaystyle \max_{\rho_N \not = \rho_0}\dim\rho\ \|\rho(D_f)\|^2 \geq \frac{m^2(n-1)}{|\widehat{K}|(|\widehat{N}|-1)}\ .
\end{equation}
\ 
\end{proof}

The proof also shows that 
\begin{equation}\label{nonabelian_bent}
\displaystyle \max_{\rho_N \not = \rho_0}\dim\rho\ \|\rho(D_f)\|^2 = \frac{m^2(n-1)}{|\widehat{K}|(|\widehat{N}|-1)} 
\quad \Leftrightarrow \quad  \forall \rho_N \not=\rho_0,\ \|\rho(D_f)\|^2=
\Gamma,
\end{equation}
where $$\Gamma=
\frac{m^2(n-1)}{\dim\rho |\widehat{K}|(|\widehat{N}|-1)}\ .$$
In the Abelian case, the righthand side of this equivalence turns 
out to be the
 definition of bent functions since 
$\dim\rho = 1$, $|\widehat{K}|=m$ and $|\widehat{N}|=n$. It is well known 
that classical bentness is equivalent to 
perfect nonlinearity, as we stated before. Now if we take the righthand side of 
equivalence~(\ref{nonabelian_bent}) as a natural definition for 
{\bf bentness} 
in the non-Abelian case, we can prove this notion to be nonequivalent to
 perfect nonlinearity in many situations. 
Let us suppose that 
$f : K \rightarrow N$ is both perfect nonlinear and bent. Since $f$ is perfect nonlinear 
its graph $D_f$, as an element of $\mathbb{Z}[G]$, 
satisfies the 
famous group ring equation (\ref{eq:rds}) for 
 relative $(m,n,m,\frac{m}{n})$. 
difference sets 
So for $(\rho,V) \in \widehat{G}$, we have
\begin{equation}
\rho(D_f D_f^{(-1)})=m\mathit{Id}_V + \lambda (\rho(G)-\rho(N))\ ,
\end{equation}
where
$\mathit{Id}_V$ denotes the identity mapping of $V$. Now let us  
suppose 
that $\rho_N \not = \rho_0$. So we have 
$\rho(G)=0_V$ and $\rho(N)=\rho_K (1_K)\otimes\rho_N (N)=0_V$ 
according to eq.~(\ref{kind_of_orthogonality_relations}). 
Then in this case we obtain:
\begin{equation}
\rho(D_f D_f^{(-1)})=\mathit \rho(D_f)\circ\rho(D_f)^*=m\mathit{Id}_V\ .
\end{equation}
By using the trace on both sides of the last equality above, we 
have for $\rho_N \not = \rho_0$
\begin{equation}\label{first_equality}
\|\rho(D_f)\|^2 = m\dim\rho .
\end{equation}
But since $f$ is bent, combining eq.~(\ref{first_equality}) and the 
righthand side of equivalence~(\ref{nonabelian_bent}),
\begin{equation}
m(n-1)=(\dim\rho)^2|\widehat{K}|(|\widehat{N}|-1)\ .
\end{equation}
This equality may hold if and only if $\dim\rho$ is the same for every $\rho 
\in\widehat{G}$ such that 
$\rho_N \not = \rho_0$ (we exclude the trivial case where $|N|=1$); this 
is for instance the case if both $K$ and $N$ are Abelian. We can show that 
in the case 
where at least one of $K$ or $N$ is non-Abelian and distinct of its derived
 group\footnote{If for a nonabelian group $H$, 
$H\not=[H,H]$ then $H$ is not a simple group.}  (for instance it is a non 
sovable group), 
and $N$, when Abelian, is not reduced to the trivial group, 
then the previous equality cannot hold and therefore 
perfect nonlinearity and bentness - as introduced by analogy - are
nonequivalent 
and even paradoxical in this 
non-Abelian setting: 
\begin{arabiclist}
\item First let us
suppose that $N$ is a non-Abelian group such that $N \not = [N,N]$ (this 
is the case of non solvable groups). By assumption $N$ has at least one
irreducible 
representation 
of dimension $d> 1$ (it is {\it ipso facto} a nonprincipal representation),
 call it $\rho^{(1)}_N$. 
Since the number of one-dimensional 
representations of $N$ is exactly $|N/[N,N]|$, there is also at least one such 
representation which is nonprincipal 
(for if $|N/[N,N]|=1$ then $N = [N,N]$). Call $\rho^{(2)}_N$ one of them.. 
Now let $\rho_K$ be any $d'$-dimensional irreducible representation of $K$,
 then $\rho_K \otimes \rho^{(i)}_N$ 
(for $i=1,2$) are both 
irreducible nonequivalent representations of $G=K \times N$ which are non
principal on $N$. We have 
$\dim \rho_K \otimes \rho_N^{(1)} = d'd > \dim\rho_K \otimes \rho^{(2)}_N =d'$.
\item Secondly, let us
suppose that $K$ is a non-Abelian group such that $K \not
  = 
[K,K]$ and, 
if $N$ is non-Abelian then $N\not=[N,N]$ and if $N$ is Abelian then it is
 not reduced to the trivial group, {\it i.e.}, $|N|>1$. 
According to case 1., we already 
know that $K$ has at least one nonprincipal one-dimensional representation,
 call it $\rho^{(1)}_K$, and at least one 
irreducible representation 
of dimension $d > 1$, call it $\rho^{(2)}_K$. If $N$ is Abelian, since it is 
not the trivial group, it has at least one 
nonprincipal (one-dimensional) representation. 
According to 1., this is also the case if $N$ is non-Abelian and $N \not = 
[N,N]$; call $\rho_N$ this representation. Then 
$\rho^{(i)}_K \otimes \rho_N$ (for $i=1,2$) are both irreducible nonequivalent 
representations of $G = K \times N$ which are 
nonprincipal on $N$. We have 
$\dim \rho^{(1)}_K \otimes \rho_N = 1 < \dim\rho^{(2)}_K \otimes \rho_N =d$.
\end{arabiclist}

\section{Some computational results and open questions}

In this section we present some results given by formal computations 
using GAP \cite{LintonGAP} and MAGMA \cite{MAGMA}.
There are three goals:
\begin{itemize}
\item {\it Almost perfect nonlinearity}, see Theorem  \ref{APN_characterization}: Minimize 
$$
\displaystyle \sum_{\rho \in \widehat{G}}\dim \rho\ \|\rho(D_f)\|^4 
$$
for functions $f:K\to N$, where $G=K\times N$.
\item {\it Maximal nonlinearity}, see eq. (\ref{eq:maxnon}): Minimize the maximum of
$$
\sqrt{\dim\rho}\ \|\rho(D_f)\|
$$
for all $f:K\to N$.
\item {\it Bentness}, see eq. (\ref{nonabelian_bent}): Find functions $f:K\to N$ such that 
$$
\forall \rho_N \not=\rho_0,\ \|\rho(D_f)\|^2=
\frac{m^2(n-1)}{\dim\rho |\widehat{K}|(|\widehat{N}|-1)}.
$$
\end{itemize}
We note that besides the bounds given in this paper,
no general results for arbitrary groups are known. 
Here we do not want to give extensive computational results but we
want to indicate some interesting phenomena which occur
in the non-Abelian setting.
Note that for groups $K$ and $N$, we have to go through
{\bf all} mappings $K\to N$, whose number  is $|N|^{|K|}$. 
Hence the approach to 
find ``good'' functions by a complete search is very limited, and 
it would be interesting to find theoretical constructions
which are {\bf provable} maximal nonlinear or almost perfect nonlinear.

Let us begin with the two goups of order $6$,
the cyclic group $\Z_6$ and the symmetric group
$\mathcal{S}_3$. In Table 1, we list the
``best'' values (marked in boldface) both for the maximum nonlinearity 
as well as 
almost perfect nonlinearity.

\begin{table}[h]
\tbl{Non-Abelian groups $K$, $N$ of order $6$, $f:K\to N$}
{\begin{tabular}{@{}cccc@{}} \toprule
$K$ & $N$ & $\displaystyle 
\min_{g:K\rightarrow N}\sum_{\rho\in\widehat{G}}\dim\ \rho\|\rho(D_g)\|^4$ & 
$\displaystyle 
\min_{g:K\rightarrow N}\max_{\rho|_N \not=\rho_0}\sqrt{\dim\ \rho}\|\rho(
D_g)\|$\\ \colrule
$\mathcal{S}_3$ & $\mathbb{Z}_6$ & {\bf{2376}} & 4 \\
$\mathbb{Z}_6$ & $\mathcal{S}_3$ & 3972 & $4\sqrt{2}$  \\
$\mathcal{S}_3$ & $\mathcal{S}_3$ & 3552 & $2\sqrt{14}$ \\
$\mathbb{Z}_6$ & $\mathbb{Z}_6$ & 2808 & ${\bf{2\sqrt{3}}}$ \\ \botrule
\end{tabular}}
\end{table}

It is  remarkable that the measure of almost
 perfect nonlinearity is by far the least in the 
case $(K,N)=(\mathcal{S}_3,\mathbb{Z}_6)$ and not in the 
Abelian case. Moreover, if $(K,N)=(\mathcal{S}_3,\mathbb{Z}_6)$
each almost perfect nonlinear functions is also maximal nonlinear. This
situation 
fails to be true in all the other cases 
where no almost perfect nonlinear function is also maximal nonlinear. 

We also made a complete search for the case of abelian groups of order $8$
(Table 2).

\begin{table}[h]
\tbl{Abelian groups $K$, $N$ of order $8$, $f:K\to N$}
{\begin{tabular}{@{}cccc@{}} \toprule
$K$ & $N$ & $\displaystyle 
\min_{g:K\rightarrow N}\sum_{\rho\in\widehat{G}}|\rho(D_g)|^4$ & 
$\displaystyle 
\min_{g:K\rightarrow N}\max_{\rho|_N \not=\rho_0}|\rho(D_g)|$\\ \colrule
$\Z_8$ & $\Z_8$ & 8832 &  $\sqrt{10+4\sqrt{2}}$\\
$\Z_8$ & $\Z_4\times \Z_2$ & 8576 & 4  \\
$\Z_8$& $\Z_2\times \Z_2\times \Z_2$ & 9216 & 4 \\
$\Z_4\times \Z_2$ & $\Z_8$ & 8960 & 4 \\
$\Z_4\times \Z_2$ & $\Z_4\times \Z_2$ & 9216 & 4  \\
$\Z_4 \times \Z_2$& $\Z_2\times \Z_2\times \Z_2$ & 10240 & 4 \\
$\Z_2\times \Z_2\times \Z_2$ & $\Z_8$ & 9216 & 4 \\
$\Z_2\times \Z_2\times \Z_2$ & $\Z_4\times \Z_2$ &  10240&  4  \\
$\Z_2\times \Z_2\times \Z_2$ & $\Z_2\times \Z_2\times \Z_2$ & 11264 & 4 \\
\botrule
\end{tabular}}\end{table}

The case $K=N=\Z_8$ is of interest: Here the maximum 
character value is strictly smaller than $4$, hence there
is a function which is ``better'' than in the elementary-Abelian
case\footnote{An example of the graph of such a function is
$\{(0,0),(1,5),(2,7),(3,7),(4,7),(5,4),(6,5),(7,4)\}$}. In contrast to the other cases, this ``smallest'' maximum
is not reached for the functions which minimize the sum of
the fourth powers of the character values (which is 8960).

We also searched  for bent functions from $\mathcal{S}_3$ to
 a group $N$ such that 
$1 < |N| \leq 5$.  Note that, in contrast to the Abelian case,
the existence question for bent functions $K\to N$ 
is also meaningful if $|N|$ is not a divisor of $|K|$.
The results that we  obtained are as follows: First of all,
there is no bent functions if $N \in \{\mathbb{Z}_2,\mathbb{Z}_4,
\mathbb{Z}_2\times\mathbb{Z}_2,
\mathbb{Z}_5\}$.
But there is (at least) one bent function $f : \mathcal{S}_3 \rightarrow 
\mathbb{Z}_3$. Indeed we can check that the map 
$f$ defined by $f(\mathit{id})=f((1,2))=f((2,3))=f((1,3))=0$ and 
$f((1,2,3))=f((1,3,2))=1$ is bent, that is to say that 
$
\sqrt{\dim \rho}\|\rho(D_f)\|=2\sqrt{3}$ for all 
$\rho|\Z_3\ne \rho_0$. The group
$\mathcal{S}_3$ has two one-dimensional  
representations and one two-dimensional representation, and $\mathbb{Z}_3$ has 
three one-dimensional representations
therefore 
$\frac{|\mathcal{S}_3|^2(|\mathbb{Z}_3|-1)}{|\widehat{\mathcal{S}_3}|(\widehat{\mathbb{Z}_3}-1)|}=12$.
We have 
$\displaystyle \sqrt{\dim \rho}\|\rho(D_f)\|=2\sqrt{3}$ for each 
representation 
$\rho \in \widehat{\mathcal{S}_3\times \mathbb{Z}_3}$ which is 
nonprincipal on $\mathbb{Z}_3$. 
Due to the fact that $\mathcal{S}_3$ has a representation of dimension 
greater than $1$, $f$ is not 
perfect nonlinear (see the discussion at the end of the previous section).

Finally, we would like to raise some questions. 
Regarding these questions,
both computational
results as well as infinite families are, of course, welcome.

\begin{enumerate}
\item Regarding the caracterisation of APN functions by mean of the sum 
$\displaystyle \sum_{\rho \in \widehat{G}}\dim \rho\ \|\rho(D_f)\|^4$ given 
in 
Theorem~\ref{APN_characterization}, 
it should be interesting to find some 
functions $f : K \rightarrow N$ for which this 
sum reaches some value which is  better than the one in the 
classical case of elementary Abelian groups. 
\item Find functions $f:K\to N$  such that 
$\dim\rho\|\rho(D_f)\|^2=
\frac{m^2(n-1)}{|\widehat{K}|(|\widehat{N}|-1)}$ for all
$\rho$ nonprincipal on $N$
{\it i.e.} non-Abelian bent functions).
\item Find functions $f : K \rightarrow N$ with $|K|=|N|=2^n=m$ ($n$ even 
so that there are no almost bent functions) such that 
$\displaystyle \max_{\rho_N \not=\rho_0}\dim\rho\|\rho(D_f)\|^2 < 4m$, that 
is to say functions which are better than 
the known maximum nonlinear functions in the elementary-Abelian case.
\item Find functions $f : K \rightarrow N$ with $|K|=|N|=2^n=m$ ($n$ odd, so 
that there are almost bent functions) such that 
$\displaystyle \max_{\rho_N \not=\rho_0}\dim\rho\|\rho(D_f)\|^2 < 2m$, 
that is to say functions which are better than 
classical almost bent functions. Note that we found such an
example for a mapping $f:\Z_8\to\Z_8$.

\end{enumerate}

\bibliographystyle{siam}
\bibliography{liter}

\end{document}